\documentclass[ejsv2,noshowframe]{imsart}

\RequirePackage[authoryear]{natbib}
\RequirePackage{mathtools}
\RequirePackage{graphicx}
\RequirePackage[colorlinks=true,citecolor=blue,linkcolor=blue,urlcolor=blue]{hyperref}

\startlocaldefs

\theoremstyle{plain}
\newtheorem{theorem}{Theorem}[section]

\newtheorem{lemma}[theorem]{Lemma}
\newtheorem{corollary}[theorem]{Corollary}

\theoremstyle{definition}
\newtheorem{remark}[theorem]{Remark}

\newcommand{\C}{\mathbb C}
\newcommand{\E}{\mathbb E}

\newcommand{\tr}{\operatorname{tr}}
\newcommand{\Var}{\operatorname{Var}}
\newcommand{\TV}{\operatorname{TV}}
\newcommand{\norm}[1]{\left\|#1\right\|}
\newcommand{\abs}[1]{\left|#1\right|}

\endlocaldefs

\begin{document}

\begin{frontmatter}

\title{Local and Global Risk Bounds for Quantum Entropy Estimation
under Projective-Design Measurements}
\runtitle{Risk bounds for quantum entropy estimation}

\begin{aug}
\author[A]{\fnms{Xinyu}~\snm{Song}
  \ead[label=e1]{song.xinyu@mail.shufe.edu.cn}}
\address[A]{School of Statistics and Data Science,
Shanghai University of Finance and Economics
\printead[presep={,\ }]{e1}}
\runauthor{X. Song}
\end{aug}

\begin{abstract}
We establish a lower bound for estimating the von Neumann entropy from
independent outcomes of any fixed rank-one POVM\@. A rotation-averaged van Trees
argument gives a global minimax risk of at least
\((d/n)\log^2\{n/(4d)\}\) when \(d\ge C\) and \(n\ge Cd\), without a projective-design
assumption. We also characterize risk on an operator-norm ball of radius
\(r\) around the maximally mixed state. We allow an approximate second moment:
on the trace-zero
Hermitian subspace, the measurement frame may differ by \(\varepsilon<1\)
from the tight projective frame. A clipped estimator based on canonical dual
shadows and the complete U-statistic for purity has risk at most
\(d^3r^2/n+d^4/n^2+d^6r^6\). Lower bounds under the same frame control yield
the local minimax rate \(d^3r^2/n+d^4/n^2\) when \(n\ge Cd^2\) and \(r\) lies
in an explicit matching range. For every fixed upper bound on
\(\varepsilon\), approximation changes
only the constants, not the powers of \(d,n,r\). At the critical radius
\(r=n^{-1/2}\), the local risk is
asymptotically negligible relative to the global risk when
\(n\log^2(n/d)\gg d^3\). This separation holds for projective 2-designs,
including global Clifford measurements in qubit dimensions, and for their
uniformly well-conditioned frame approximations.
Finite-sample experiments in dimension four illustrate the critical-radius
benchmark and the effect of a nonexact frame.
\end{abstract}



\end{frontmatter}

\section{Introduction}

Let \(\rho\) be the density matrix of a quantum system on \(\C^d\). Its von
Neumann entropy is
\[
 S(\rho)=-\tr(\rho\log\rho),
\]
where \(0\log 0=0\). The statistical experiment is determined by the manner in
which repeated preparations of the state are measured. Collective procedures
may act jointly on several copies and include methods based on Schur--Weyl
duality and empirical Young diagrams
\citep{ODonnellWright2016,AISW2020,WangZhang2025}. By comparison,
randomized-measurement and classical-shadow protocols measure one copy at a
time and retain a classical record \citep{HKP2020,ElbenReview2023}. A fixed
single-copy measurement generally contains less information than a collective
measurement. Risk bounds obtained under one experiment therefore need not
hold under the other. Statistical formulations of quantum measurement and
state estimation are reviewed in \citet{WangSong2020}; lower bounds for
restricted measurements are studied by
\citet{ChenHuangLiLiu2022,LoweNayak2025}.

Entropy estimation is a nonregular functional problem even in classical
models. For a multinomial distribution with a large alphabet, small cell
probabilities govern the minimax risk, and polynomial-approximation estimators
can substantially improve on the empirical plug-in estimator
\citep{JiaoVenkatHan2015,WuYang2016}. The projective-design experiment studied
here is different: its observations are measurement outcomes, not draws from
the eigenvalue distribution of \(\rho\). Both the eigenvalues and the
eigenvectors enter the observation law through the measurement channel.
Consequently, the classical entropy rates do not determine the risk under a
fixed randomized measurement.

Quantum tomography provides another natural comparison. A tomographic
procedure reconstructs the full matrix and then applies the entropy
functional, whereas the local estimator below targets only the quadratic
component that is statistically visible at \(I/d\). Under the same fixed
rank-one measurement, the empirical mean of calibrated dual shadows is a
linear-inversion state estimate; it may be projected onto the state space and
used as an entropy plug-in rule. Such a rule must reconstruct matrix
directions not targeted by the local quadratic functional, and the
nonsmooth entropy map can amplify reconstruction error near the boundary.
Our results neither give a tomography rate nor claim uniform domination of
tomographic plug-in estimators. The numerical comparison below instead
evaluates direct functional estimation and a dual-shadow tomography baseline
under an identical sample budget and observation model. General minimax
tomography under sparse Pauli measurements is studied by
\citet{CaiEtAl2016}, while frame-based connections between tomography and
shadow estimation are developed by \citet{InnocentiEtAl2023}.

The behavior of entropy also changes across the quantum state space. At the
maximally mixed state \(I/d\), its first derivative vanishes in every
trace-zero direction. On an operator-norm ball of radius \(r\) around \(I/d\),
the leading variation is therefore quadratic and is determined by the purity
\(\tr(\rho^2)\). Near the boundary of the state space, the derivative instead
increases as an eigenvalue approaches zero. These observations lead to the two
questions considered in this paper: how does the minimax risk depend on the
radius of a neighborhood of \(I/d\), and how does this local risk compare with
the risk over the full state space?

The center \(I/d\) is not chosen as a surrogate for an unknown state. It is
the unique unitarily invariant state and the stationary point of entropy under
the trace constraint. The class center and radius are known parts of the
statistical problem, while the state inside the ball remains unknown. Varying
\(r\) measures how the problem changes as one moves away from this
second-order point: at \(r\asymp n^{-1/2}\), the entropy diameter is already
attained by the constant rule \(\log d\), whereas on larger balls data can
improve the risk. The global class serves a different purpose. It includes
states near the boundary, where entropy is nonsmooth, and asks whether the
regular local benchmark accounts for the worst-case difficulty. A comparison
of the two risks is therefore a statement about spatial inhomogeneity of the
same experiment, not a device for replacing the global problem by a known
local model.

Complete pair averages are commonly used to estimate purity and related
multi-copy observables from randomized measurements
\citep{VanEnkBeenakker2012,HKP2020,VermerschEtAl2024}. Recent work has also
considered measurement design for nonlinear properties such as
\(\tr(O\rho^2)\) \citep{DuEtAl2026}, collision-based single-copy measurements
for higher trace powers \citep{LiEtAl2026}, and joint-copy procedures for
polynomial spectral functionals \citep{ChenWangYuZhang2026,YaoEtAl2026}.
Analytic continuation of randomized-measurement estimates of R\'enyi entropies
provides another route to the von Neumann entropy \citep{VijayEtAl2026}.
These methods use different observation models or address algorithmic recovery
rather than minimax risk over radius-indexed classes. Local asymptotic minimax
theory for quantum state estimation has also been developed under collective
or adaptive measurements \citep{GutaJanssens2008}, but those results concern a
different loss and a different statistical experiment.

Two related manuscripts by the author supply ingredients but not the present
theorems. \citet{SongTracePolynomial2026} develops Hoeffding decompositions
and compares batched with complete U-statistics for trace polynomials of a
fixed projected block under global Clifford shadows. Its entropy application
is an oracle small-spectrum functional; it does not estimate the full von
Neumann entropy, define a local parameter class, or prove minimax lower
bounds. The present paper imports complete symmetrization at degree two, but
adds the radius-dependent entropy reduction, matching local lower bounds, and
the exact-to-approximate frame analysis.

\citet{SongEntropyLowerBound2026} proves
\(R_{n,d}^{\mathrm{glob}}\gtrsim\log^2(d^2/n)\) for
\(d\lesssim n\lesssim d^2\) under an exact projective 2-design. Its proof
uses the exact design affinity to compare the maximally mixed state with a
Haar mixture of low-rank states. That identity is unavailable for a general
fixed POVM and cannot yield the global theorem below. Here the global lower
bound applies to every fixed rank-one POVM, uses a deterministic hard
orientation selected by rotation-averaged Fisher information, and holds for
all \(d\ge C\) and \(n\ge Cd\), overlapping the earlier window and extending beyond
\(d^2\). The distinction is the arbitrary fixed rank-one POVM and the
rotation-averaged van Trees interface, not a uniformly stronger lower bound.
The local \(d^4/n^2\) lower bound also uses a
different pair of full-rank, infinitesimally separated spectral mixtures.
These differences prevent either earlier theorem from implying the present
results.

The principal result of this paper is a global lower bound that applies to an
arbitrary fixed rank-one POVM\@. For each such measurement, we construct a
spectral submodel that transfers probability mass between two blocks of
comparable rank. A direct Fisher-information calculation for a fixed spectral
orientation depends on the particular measurement through inverse outcome
probabilities. We instead average the orientation of the submodel over the
unitary group. The average Fisher information coincides with its Haar value,
which can be evaluated through a Beta distribution, and therefore one
deterministic orientation has no larger prior-averaged information. A van
Trees argument then gives
\[
 R_{n,d}^{\mathrm{glob}}(\nu)
 \gtrsim
 \frac d n\log^2\left(\frac n{4d}\right),
 \qquad d\ge C,\quad n\ge Cd.
\]
The measurement remains fixed throughout this argument; only the hard
submodel is oriented relative to it. Thus the result holds for every fixed
rank-one POVM satisfying the usual normalization, without a design condition.

We use a local analysis to determine whether the global difficulty is visible
in a regular neighborhood of the maximally mixed state. Let
\(\Theta_d^{\mathrm{loc}}(r)\) be the operator-norm ball of radius \(r\) centered
at \(I/d\). The first entropy derivative vanishes on trace-zero perturbations
at this state, so the leading term is quadratic and can be estimated through
purity. We require only that the second-moment measurement frame be within a
fixed relative error \(\varepsilon<1\) of a tight projective frame. Substituting
the complete U-statistic formed from calibrated dual shadows into the
quadratic expansion gives
\[
 \sup_{\rho\in\Theta_d^{\mathrm{loc}}(r)}
 \E_\rho\{\widehat S_{\mathrm{loc}}-S(\rho)\}^2
 \lesssim
 \frac{d^3r^2}{n}+\frac{d^4}{n^2}+d^6r^6.
\]
The estimator uses only the observed outcomes and the known measurement
frame; it requires no population spectral projector, eigengap, or pilot
sample. A one-dimensional spectral submodel and a two-prior mixture
construction give lower bounds for the first two terms under the same frame
control. It follows that
\[
 R_{n,d}^{\mathrm{loc}}(r;\nu)
 \asymp
 \frac{d^3r^2}{n}+\frac{d^4}{n^2}
\]
when \(n\ge Cd^2\) and
\[
 n^{-1/2}\le r\le
 c\max\{(dn)^{-1/3},(d^3n)^{-1/4}\}.
\]
At the critical radius \(r=n^{-1/2}\), the constant estimator \(\log d\)
attains the same order. On larger balls within the matching range, the
complete-U estimator has asymptotically smaller risk than this constant-rule
benchmark whenever \(r\sqrt n\to\infty\). The result applies to the Haar
covariant measurement in every dimension, to the global Clifford measurement
in qubit dimensions, and to their uniformly conditioned frame
approximations. For fixed \(\varepsilon\), approximation changes constants but
not the minimax rate.

For every measurement covered by the local frame condition, the global lower
bound and the local upper bound concern the same experiment. If
\(n\log^2(n/d)\gg d^3\), the minimax risk over an
\(n^{-1/2}\) neighborhood of \(I/d\) is negligible relative to the global
risk. This comparison is a separation result, not a characterization of the
global minimax rate. The tight local bound supplies the benchmark needed for
this comparison and shows that behavior near \(I/d\) does not account for the
global difficulty witnessed by the small-eigenvalue submodel.

The rest of the paper proceeds as follows.
Section~\ref{sec:model} introduces the observation model and the local and
global risks. Section~\ref{sec:local} develops the local benchmark near the
maximally mixed state. Section~\ref{sec:numerics} gives a finite-sample
illustration. Section~\ref{sec:global} establishes the global lower bound and
the local--global separation. Section~\ref{sec:conclusion} concludes the
paper.
Proofs are collected in the appendices.

\section{Observation model and risks}
\label{sec:model}

Let \(\Theta_d\) denote the set of density matrices on \(\C^d\).
For \(z\) in a measurable space \(\mathcal Z\), let \(P_z\) be a rank-one
projector and let \(\nu\) be a probability measure on \(\mathcal Z\)
satisfying
\[
 \int P_z\,\nu(dz)=\frac Id.
\]
Then \(M(dz)=dP_z\,\nu(dz)\) is a rank-one POVM\@. Conversely, every rank-one
POVM admits such a representation after normalizing the trace measure.
The measure \(\nu\) is called an exact complex projective \(t\)-design if, for
\(1\le k\le t\),
\begin{equation}
 \int P_z^{\otimes k}\,\nu(dz)
 =
 \frac{1}{d(d+1)\cdots(d+k-1)}
 \sum_{\pi\in\mathfrak S_k}W_\pi,
\label{eq:design}
\end{equation}
where \(W_\pi\) denotes the operator that permutes the \(k\) tensor factors
according to \(\pi\). If the system is in state \(\rho\), one observation
\(Z\) has density
\begin{equation}
 p_\rho(z)=d\,\tr(\rho P_z)
\label{eq:density}
\end{equation}
with respect to \(\nu\). We will also allow a controlled departure from the
exact second moment. Let \(\mathbb H_{d,0}\) be the real Hilbert space of
traceless Hermitian matrices, equipped with
\(\langle A,B\rangle_F=\tr(AB)\) and
\(\norm{A}_F=\{\tr(A^2)\}^{1/2}\), and
define the measurement-frame operator
\begin{equation}
 \mathcal F_\nu(A)
 =d\int\tr(AP_z)P_z\,\nu(dz).
\label{eq:frame-operator}
\end{equation}
The first-moment normalization implies that \(\mathcal F_\nu(I)=I\) and that
\(\mathbb H_{d,0}\) is invariant under \(\mathcal F_\nu\). Write
\(\mathcal F_{\nu,0}\) for the restriction to this subspace. We say that
\(\nu\) is a \emph{frame-\(\varepsilon\) projective 2-design} if
\begin{equation}
 \norm{(d+1)\mathcal F_{\nu,0}-I_{\mathbb H_{d,0}}}_{2\to2}
 \le\varepsilon,
 \qquad 0\le\varepsilon<1,
\label{eq:approx-frame}
\end{equation}
where \(2\to2\) denotes the superoperator norm induced by the Frobenius
norm. This is the second-moment metric needed by the statistical argument.
The case \(\varepsilon=0\) is equivalent to an exact projective 2-design.
For \(\varepsilon<1\), the frame operator is invertible and the canonical
dual shadow
\begin{equation}
 X_\nu=\mathcal F_\nu^{-1}(P_Z)
 =\frac Id+\mathcal F_{\nu,0}^{-1}\left(P_Z-\frac Id\right)
\label{eq:dual-shadow}
\end{equation}
satisfies \(\E_\rho X_\nu=\rho\) and \(\tr(X_\nu)=1\). Under an exact design,
\eqref{eq:dual-shadow} reduces to \((d+1)P_Z-I\).
Dual-frame reconstructions for general informationally complete measurements
are discussed by \citet{InnocentiEtAl2023} and \citet{FischerEtAl2024}.
Throughout, \(Z_1,\ldots,Z_n\) are independent observations with joint law
\((p_\rho\nu)^{\otimes n}\), and \(\E_\rho\) denotes expectation under this
product experiment.

Two examples will be used below. Haar measure on pure states is an exact
projective design of every order. If \(d=2^q\), the uniform measure on
stabilizer states is an exact projective 3-design, and hence a 2-design
\citep{Webb2016,Zhu2017}. The latter POVM experiment is equivalent to the
global Clifford protocol. To see this, let \(U\) be uniform on the Clifford
group modulo phases and, conditional on \(U\), record the computational-basis
outcome \(b\) after applying \(U\) to the state. Relative to the uniform
measure on \((U,b)\), the joint density is
\[
 d\,\tr(\rho P_{U,b}),
 \qquad P_{U,b}=U^*|b\rangle\langle b|U.
\]
The Clifford group acts transitively on the stabilizer states, so the map
\((U,b)\mapsto P_{U,b}\) has fibers of equal cardinality. It follows that
\(P_{U,b}\) is uniformly distributed on the stabilizer states. Conditional on
this projector, the record \((U,b)\) is uniform on the corresponding fiber and
does not depend on \(\rho\). Hence \(P_{U,b}\) is sufficient, and the two
experiments are statistically equivalent.

Condition~\eqref{eq:approx-frame} also covers nonexact measurements. For
example, let \(\nu_{\rm H}\) be the Haar projective measure and let
\(\nu_{\rm B}\) be uniform on the projectors of one orthonormal basis. For
\(0\le\delta<1\), the mixture
\begin{equation}
 \nu_\delta=(1-\delta)\nu_{\rm H}+\delta\nu_{\rm B}
\label{eq:haar-basis-mixture}
\end{equation}
has the exact first moment and is frame-\(\varepsilon\) with
\(\varepsilon=d\delta\), provided \(d\delta<1\). Indeed, on
\(\mathbb H_{d,0}\), its frame has eigenvalue
\((1-\delta)/(d+1)\) on the off-diagonal subspace and
\((1+d\delta)/(d+1)\) on the diagonal subspace. Thus the class is not merely
a restatement of exact designs. More generally, sampling every outcome of a
random orthonormal basis preserves the first moment exactly, while the
second-frame error quantifies how far the resulting measurement is from a
tight projective frame.

For \(r>0\), define the local class
\begin{equation}
 \Theta_d^{\mathrm{loc}}(r)
 =
 \left\{
 \rho\in\Theta_d:
 \norm{\rho-I/d}_{\mathrm{op}}\le r
 \right\}.
\label{eq:local-class}
\end{equation}
Let \(\mathcal E_n\) be the set of measurable real-valued functions of
\(Z_1,\ldots,Z_n\). The local and global minimax risks are
\begin{align}
 R_{n,d}^{\mathrm{loc}}(r;\nu)
 &=
 \inf_{\widehat S\in\mathcal E_n}
 \sup_{\rho\in\Theta_d^{\mathrm{loc}}(r)}
 \E_\rho\{\widehat S-S(\rho)\}^2,
\label{eq:local-risk}\\
 R_{n,d}^{\mathrm{glob}}(\nu)
 &=
 \inf_{\widehat S\in\mathcal E_n}
 \sup_{\rho\in\Theta_d}
 \E_\rho\{\widehat S-S(\rho)\}^2.
\label{eq:global-risk}
\end{align}
The two risks concern the same fixed single-copy experiment and differ only
in their parameter spaces.

\section{A local benchmark near the maximally mixed state}
\label{sec:local}

Let \(X_i=\mathcal F_\nu^{-1}(P_{Z_i})\), \(1\le i\le n\), be the canonical
dual shadows in \eqref{eq:dual-shadow}. We estimate the purity by the complete
degree-two U-statistic
\begin{equation}
 \widehat T_2^U
 =
 \binom n2^{-1}
 \sum_{1\le i<j\le n}\tr(X_iX_j).
\label{eq:purity-u}
\end{equation}
This statistic is unbiased for \(\tr(\rho^2)\). Let
\(\Pi_{[0,\log d]}(x)=\min\{\log d,\max(0,x)\}\), and define
\begin{equation}
 \widehat S_{\mathrm{loc}}
 =
 \Pi_{[0,\log d]}
 \left[
 \log d-\frac d2
 \left\{\widehat T_2^U-\frac1d\right\}
 \right].
\label{eq:local-estimator}
\end{equation}
The expression inside the projection is the quadratic Taylor approximation to
entropy at \(I/d\), with purity replaced by \eqref{eq:purity-u}. The following
theorem gives its uniform risk and a corresponding minimax lower bound.

\begin{theorem}[Radius-dependent local bounds]
\label{thm:local}
Fix \(0\le\varepsilon_0<1\). There are constants \(c>0\) and
\(C_{\varepsilon_0}>0\), where \(c\) is universal and
\(C_{\varepsilon_0}\) depends only on \(\varepsilon_0\), such that the
following statements hold.

\noindent{\rm (i)}
If \(\nu\) is a frame-\(\varepsilon\) projective 2-design with
\(\varepsilon\le\varepsilon_0\), then, for every \(d\ge2\), \(n\ge2\), and
\(0<r\le1/(4d)\),
\[
 \sup_{\rho\in\Theta_d^{\mathrm{loc}}(r)}
 \E_\rho\{\widehat S_{\mathrm{loc}}-S(\rho)\}^2
 \le C_{\varepsilon_0}\left(
 \frac{d^3r^2}{n}+\frac{d^4}{n^2}+d^6r^6
 \right).
\]

\noindent{\rm (ii)}
If \(\nu\) is a frame-\(\varepsilon\) projective 2-design with
\(\varepsilon\le\varepsilon_0\), then, for every \(d\ge2\), \(n\ge2\), and
\(n^{-1/2}\le r\le1/(4d)\),
\[
 R_{n,d}^{\mathrm{loc}}(r;\nu)
 \ge c\left(
 \frac{d^3r^2}{n}+\frac{d^4}{n^2}
 \right).
\]
\end{theorem}

On a suitable range of radii, the Taylor remainder in part (i) is no larger
than the two stochastic terms. The upper and lower bounds then have the same
order.

\begin{corollary}[Local minimax rate]
\label{cor:local-rate}
For every fixed \(0\le\varepsilon_0<1\), there are positive constants
\(c_0,c,C\), depending at most on \(\varepsilon_0\), such that, if \(\nu\)
is a frame-\(\varepsilon\) projective 2-design with
\(\varepsilon\le\varepsilon_0\), \(d\ge2\), \(n\ge Cd^2\), and
\[
 n^{-1/2}\le r\le
 c_0\max\{(dn)^{-1/3},(d^3n)^{-1/4}\},
\]
then
\begin{equation}
 R_{n,d}^{\mathrm{loc}}(r;\nu)
 \asymp_{\varepsilon_0}
 \frac{d^3r^2}{n}+\frac{d^4}{n^2}.
\label{eq:local-rate}
\end{equation}
\end{corollary}

\begin{proof}
Choose \(c_0\) sufficiently small and \(C\) sufficiently large. Since
\(n\ge Cd^2\), the interval of radii is nonempty and its upper endpoint is at
most \(1/(4d)\). If the maximum in the upper endpoint is \((dn)^{-1/3}\),
then
\[
 d^6r^6\le c_0^6\frac{d^4}{n^2}.
\]
If the maximum is \((d^3n)^{-1/4}\), then
\[
 d^6r^6\le c_0^4\frac{d^3r^2}{n}.
\]
Thus the Taylor remainder is bounded by a universal constant times the sum of
the two terms in \eqref{eq:local-rate}.
The conclusion follows by applying the two parts of
Theorem~\ref{thm:local} to the same frame-approximate experiment.
\end{proof}

The upper endpoint in Corollary~\ref{cor:local-rate} is \((dn)^{-1/3}\) when
\(n\le d^5\) and \((d^3n)^{-1/4}\) when \(n\ge d^5\). The first branch makes
the interval nonempty once \(n\) is of order \(d^2\); the second allows a
larger radius when the sample size exceeds order \(d^5\). The two terms in
\eqref{eq:local-rate} are equal at \(r=\sqrt{d/n}\). Hence
\(d^3r^2/n\) is the leading term within the matching range only when
\(n\gtrsim d^5\) and \(r\) lies above this crossover. Below the crossover,
\(d^4/n^2\) is the larger term. This comparison concerns the decomposition of
the minimax rate and does not, by itself, establish optimality of a constant
estimator.

Exact first-moment normalization together with the two-sided frame bound
\eqref{eq:approx-frame} suffices for both bounds. On the local class, the
one-observation density relative to the design measure is uniformly bounded
between \(3/4\) and \(5/4\). Combining this likelihood-ratio bound with the
inverse-frame bound gives a dimension-free covariance bound for a single dual
shadow on the trace-zero Hermitian subspace, while complete symmetrization
assigns an \(n^{-2}\) factor to the degenerate second Hoeffding component. For
the lower bounds, a fixed spectral direction produces \(d^3r^2/n\), and two
fixed spectra with Haar-distributed eigenbases produce \(d^4/n^2\). The
second construction must account for the lack of unitary covariance of an
approximate frame; a two-rotation kernel and unitary concentration provide a
dimension-free replacement. The required frame identities are stated in
Appendix~\ref{app:design}; the upper and lower bounds are proved in
Appendices~\ref{app:local-upper} and~\ref{app:local-lower}, respectively.

\begin{remark}[What is robust to design approximation]
\label{rem:approx-scope}
The extension is deliberately stated in the induced frame norm used by the
proof. The entropy Taylor expansion, the complete-U Hoeffding decomposition,
and the global lower bound are unchanged. Three exact-design interfaces are
modified: the nominal shadow is replaced by the calibrated dual shadow; the
local covariance and fixed-direction Fisher information use the upper and
lower frame eigenvalues; and the rotated-mixture affinity uses two independent
rotations rather than a kernel depending only on their relative rotation. No
third or higher design moment is needed.

Approximate quantum-state and unitary designs are formulated through
different operational comparison criteria in the literature
\citep{AmbainisEmerson2007,DankertEtAl2009}. Such a certificate implies
\eqref{eq:approx-frame} only after an explicit, possibly
dimension-dependent, transfer to the induced projective measurement frame.
Theorem~\ref{thm:local} does
not cover an approximate first moment, an unknown frame that must itself be
estimated, or a sequence with \(\varepsilon\uparrow1\). In the last case the
dual-frame constants diverge; if the frame loses a traceless direction, that
direction is not identifiable. Thus a fixed frame error below one affects
only constants, whereas deteriorating conditioning can affect the effective
rate or invalidate estimation altogether.
\end{remark}

\begin{remark}[The critical radius]
When \(dr\le1/4\), the entropy diameter of
\(\Theta_d^{\mathrm{loc}}(r)\) is of order \(d^2r^2\). Hence the constant
estimator \(\widehat S=\log d\) has worst-case squared error of order
\(d^4r^4\). The upper bound follows from the same Taylor expansion used in
the proof of Theorem~\ref{thm:local}(i), and perturbations with eigenvalues
\(+r\) and \(-r\) on two blocks of comparable rank give the reverse bound.
At \(r=n^{-1/2}\), this benchmark has order \(d^4/n^2\), which agrees with the
endpoint of \eqref{eq:local-rate}. Thus data are not needed to attain the
minimax order at the critical radius. If \(r\sqrt n\to\infty\) within the
range of Corollary~\ref{cor:local-rate}, then
\[
 \frac{d^3r^2/n+d^4/n^2}{d^4r^4}
 =\frac1{dnr^2}+\frac1{n^2r^4}\longrightarrow0.
\]
Therefore, on the larger local balls covered by the corollary, the complete-U
estimator has asymptotically smaller worst-case risk than the constant rule.
This improvement may occur when \(d^2\ll n\ll d^5\), even though
\(d^4/n^2\) remains the larger term in \eqref{eq:local-rate} throughout that
sample range.
\end{remark}

\begin{remark}
The estimator in \eqref{eq:local-estimator} does not use a population spectral
decomposition or class-level eigengap information. It does, however, exploit
the fact that the class in \eqref{eq:local-class} is centered at the known
state \(I/d\). Theorem~\ref{thm:local} does not extend directly to a
neighborhood of an unknown state, where the first derivative of entropy need
not vanish. For a nonexact frame, the estimator also requires calibration of
the known measurement operator \(\mathcal F_\nu\). Using the exact-design
formula \((d+1)P_Z-I\) without this correction would generally introduce
design-dependent bias.
\end{remark}

\section{Finite-sample illustration}
\label{sec:numerics}

We illustrate the two radius regimes and the effect of a nonexact frame in
dimension \(d=4\). Let \(U\) be one fixed Haar rotation and consider
\[
 \rho_n=\frac I4+r_nU\operatorname{diag}(1,1,-1,-1)U^*.
\]
The critical path uses \(r_n=n^{-1/2}\), while the supercritical diagnostic
uses \(r_n=n^{-2/5}\). We take \(n\in\{1024,2048,4096,8192\}\); every state
satisfies \(4r_n\le1/4\). The measurements are the Haar--basis mixtures in
\eqref{eq:haar-basis-mixture}, with \(\delta=\varepsilon/4\) and
\(\varepsilon\in\{0,0.5\}\). Thus \(\varepsilon=0\) is the exact Haar
measurement, whereas \(\varepsilon=0.5\) is a nonexact, well-conditioned
frame. Outcomes from both mixture components are sampled exactly, and all
estimators use the calibrated dual shadow in \eqref{eq:dual-shadow}.

Figure~\ref{fig:local-numerics} compares the complete-U estimator, the
constant rule \(\log d\), and a tomography plug-in baseline. The latter is the
empirical mean of the dual shadows, projected in Frobenius norm onto the
density-matrix simplex before applying entropy. Each point is based on 2,000
independent repetitions; the error bars are \(1.96\) Monte Carlo standard
errors of the estimated MSE. At the critical path, the endpoint log--log
slopes of the complete-U MSE are \(-1.971\) for \(\varepsilon=0\) and
\(-1.947\) for \(\varepsilon=0.5\), close to the \(n^{-2}\) scale, while the
constant rule has the smaller finite-sample constant. This agrees with the
critical-radius conclusion that data are not needed to attain the minimax
order.
The reported slopes concern the clipped estimator in
\eqref{eq:local-estimator}: its clipping rate ranges from \(31.25\%\) to
\(37.75\%\) on the critical grid, compared with \(1.30\%\) to \(6.85\%\) on
the supercritical grid. The tomography projection is inactive in all released
repetitions, but it is retained to define a valid state estimator outside this
grid.

On the supercritical path, the corresponding slopes are \(-1.810\) and
\(-1.812\), close to the \(n^{-9/5}\) first stochastic term obtained by
substituting \(r_n=n^{-2/5}\) into \(d^3r_n^2/n\). At \(n=8192\), the
complete-U, constant, and tomography MSEs are respectively
\(1.03\times10^{-5}\), \(3.52\times10^{-5}\), and
\(3.21\times10^{-5}\) for the exact frame; for \(\varepsilon=0.5\), they are
\(1.09\times10^{-5}\), \(3.52\times10^{-5}\), and
\(3.42\times10^{-5}\). The frame perturbation therefore produces a moderate
constant change on this grid without a visible change of scaling.

\begin{figure}[t]
\centering
\includegraphics[width=\linewidth]{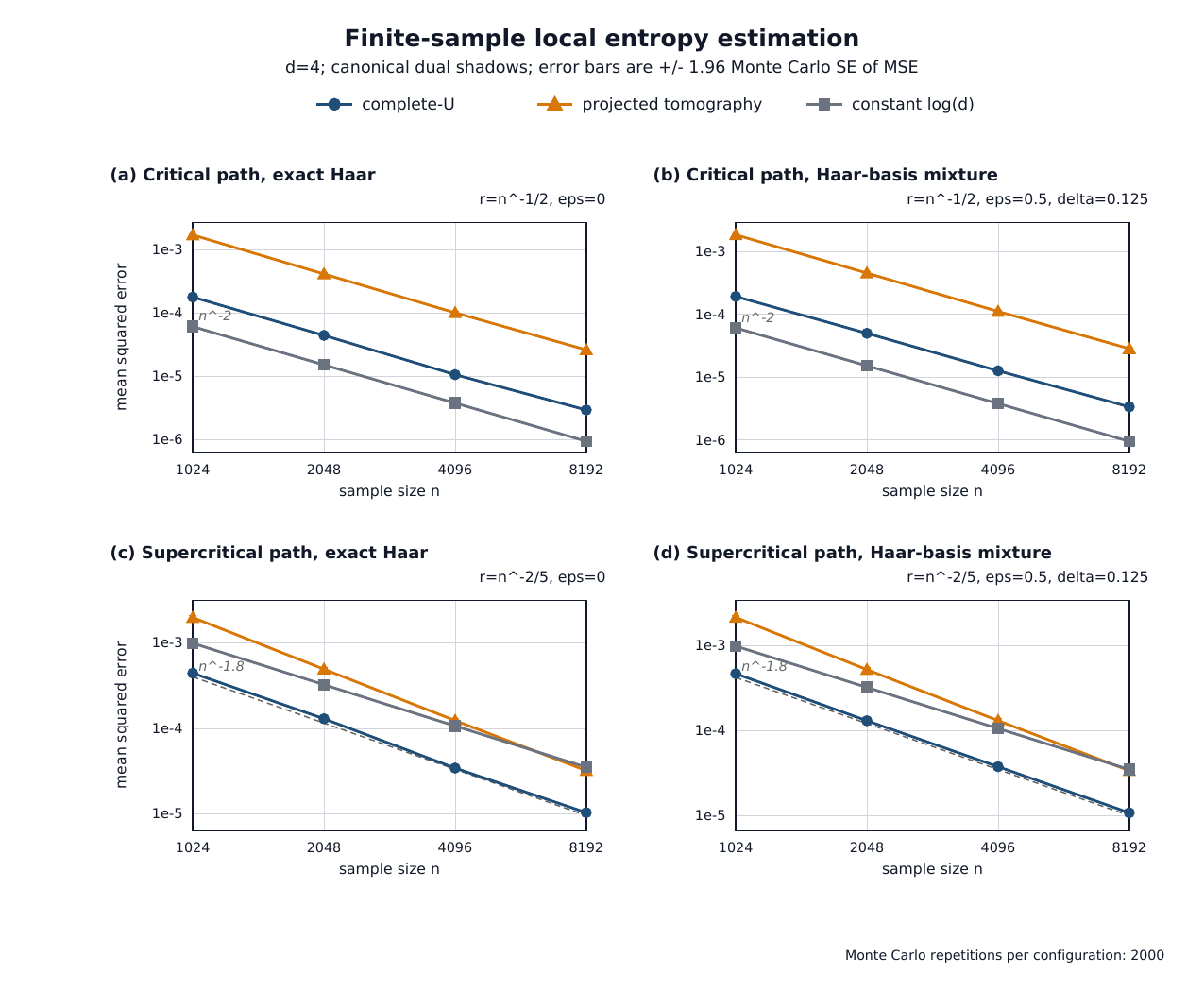}
\caption{Finite-sample MSE under exact Haar and nonexact Haar--basis
measurements. Panels (a)--(b) use \(r_n=n^{-1/2}\); panels (c)--(d) use
\(r_n=n^{-2/5}\). Error bars are \(1.96\) Monte Carlo standard errors based on
2,000 repetitions. Reference lines show the leading powers discussed in the
text.}
\label{fig:local-numerics}
\end{figure}

The experiment is an illustration, not a verification of minimax quantifiers.
It uses one orientation and one dimension, and the unknown constant in the
upper radius of Corollary~\ref{cor:local-rate} is not calibrated by the
simulation. The released code, configuration, raw replicate data, summaries,
and figure use fixed seeds and reproduce the reported values from the stated
measurement law.

\section{Global lower bound and local--global separation}
\label{sec:global}

We next consider the global risk in \eqref{eq:global-risk}. The measurement
\(\nu\) is an arbitrary fixed rank-one POVM satisfying the normalization in
Section~\ref{sec:model}. The lower bound is obtained from a spectral submodel
that transfers probability mass between two blocks of comparable rank. Its
orientation is chosen relative to the fixed measurement. Along this submodel,
the entropy derivative increases as the mass of the smaller block decreases,
whereas the prior-averaged Fisher information from one observation remains of
order \(d^{-1}\).

\begin{theorem}[Global rank-one-measurement lower bound]
\label{thm:global}
There are universal constants \(c,C>0\) such that, for every \(d\ge C\),
every \(n\ge Cd\), and every fixed rank-one POVM \(\nu\) on \(\C^d\),
\begin{equation}
 R_{n,d}^{\mathrm{glob}}(\nu)
 \ge
 c\,\frac d n\log^2\left(\frac n{4d}\right).
\label{eq:global-lower}
\end{equation}
\end{theorem}

The argument uses the one-dimensional van Trees inequality
\citep{GillLevit1995}. For even \(d\), let \(\Pi_s\) and \(\Pi_b\) project
onto two orthogonal subspaces of rank \(d/2\), and consider
\begin{equation}
 \rho_t=
 \frac{2(\mu+t)}d\Pi_s+
 \frac{2(1-\mu-t)}d\Pi_b,
 \qquad |t|\le\frac\mu2.
\label{eq:global-submodel}
\end{equation}
For odd \(d\), two rank-\((d-1)/2\) blocks are supplemented by a
one-dimensional block with eigenvalue \(1/d\). Taking
\(\mu=2\sqrt{d/n}\), the entropy derivative along this family is bounded
below by a constant multiple of \(\log(1/\mu)\).

For a Haar direction \(u\), its squared projection onto a fixed rank-\(d/2\)
subspace has a Beta distribution. This identity gives an order-\(d^{-1}\)
bound for the Haar Fisher information. Averaging
\eqref{eq:global-submodel} over its spectral orientation gives the same
prior-averaged information for any fixed rank-one POVM\@. One may therefore
select a deterministic orientation with Fisher information no larger than
the Haar average and then apply van Trees. The odd-dimensional construction
and the full proof are given in Appendix~\ref{app:global-proof}.

\begin{remark}[Role of rotation averaging]
Theorem~\ref{thm:global} is not obtained by comparing the fixed measurement
with the Haar experiment. For each \(\nu\), the argument selects a difficult
spectral orientation within the parameter space of the same experiment. The
Beta distribution enters only through the average over orientations; the
proof does not approximate an inverse measurement probability by a finite
design.
\end{remark}

If \(\nu\) obeys the approximate frame condition, the local upper bound and
the global lower bound concern the same observation model. Their comparison
gives the following separation result.

\begin{corollary}[Local--global separation]
\label{cor:separation}
Fix \(0\le\varepsilon_0<1\). There are constants \(c,C>0\), depending at
most on \(\varepsilon_0\), such that, if \(\nu\) is a
frame-\(\varepsilon\) projective 2-design with
\(\varepsilon\le\varepsilon_0\), \(d\ge C\), and \(n\ge Cd^2\), then
\[
 \frac{R_{n,d}^{\mathrm{glob}}(\nu)}
 {R_{n,d}^{\mathrm{loc}}(n^{-1/2};\nu)}
 \ge
 c\frac{n}{d^3}
 \log^2\left(\frac n{4d}\right).
\]
Hence, along any sequence satisfying
\[
 \frac{n}{d^3}\log^2\left(\frac n{4d}\right)\longrightarrow\infty,
\]
the local minimax risk is asymptotically negligible relative to the global
minimax risk.
\end{corollary}

\begin{proof}
Apply the lower bound in Theorem~\ref{thm:global} to the numerator and the
upper bound in Theorem~\ref{thm:local} with \(r=n^{-1/2}\) to the
denominator. For \(n\ge Cd^2\), the latter is at most a constant depending
only on \(\varepsilon_0\) times \(d^4/n^2\).
\end{proof}

Corollary~\ref{cor:separation} does not determine the global minimax rate. It
shows that, along the sequences specified in the corollary, the regular local
problem around \(I/d\) has a smaller risk scale than the full problem.
The corollary applies to the Haar covariant measurement in every dimension,
to the global Clifford measurement in qubit dimensions, and to uniformly
well-conditioned frame approximations of these experiments. In the submodel
\eqref{eq:global-submodel}, one half of the spectrum has total mass
\(\mu=2\sqrt{d/n}\). In the separation regime, the corresponding small
eigenvalues lie outside the \(n^{-1/2}\) operator-norm neighborhood of
\(I/d\).

\section{Conclusion}
\label{sec:conclusion}

Over the full state space, the rotation-averaged construction in this paper
gives a lower bound of order
\[
 \frac dn\log^2\left(\frac n{4d}\right)
\]
for every fixed rank-one POVM when \(d\ge C\) and \(n\ge Cd\). The proof selects a difficult spectral
orientation relative to the fixed measurement by averaging its
prior-weighted Fisher information over unitary rotations. It therefore does
not require the measurement to be a projective design.

The local analysis provides a tight benchmark for comparison with this global
bound. On the range in Corollary~\ref{cor:local-rate}, the minimax risk over an
operator-norm ball of radius \(r\) around \(I/d\) is
\(d^3r^2/n+d^4/n^2\). The upper bound is attained by a complete U-statistic for
purity formed from canonical dual shadows. Exact designs are not required:
for every fixed frame error below one, the upper and lower bounds retain the
same powers of \(d,n,r\), with only their constants changing. At the critical
radius \(n^{-1/2}\), the constant estimator has risk of order \(d^4/n^2\);
on larger balls with
\(r\sqrt n\to\infty\), the data-dependent rule improves on this benchmark.
The finite-sample experiment illustrates this transition and the moderate
constant effect of one nonexact frame; it is not evidence of uniformity over
states, dimensions, or approximate-design constructions.
For these exact and approximate projective frames, the local and global
results yield a separation between the regular behavior near \(I/d\) and the
global difficulty associated
with small eigenvalues. A matching global upper bound remains open. Its
construction would require an estimator that accommodates both the quadratic
behavior near \(I/d\) and the nonsmooth behavior near the boundary of the
state space.

\begin{appendix}

\section{Measurement-frame identities}
\label{app:design}

\begin{lemma}
\label{lem:design-identities}
Suppose that \(\nu\) has the first-moment normalization in
Section~\ref{sec:model}. The operator \(\mathcal F_\nu\) is self-adjoint and
positive semidefinite on the real Hilbert space of Hermitian matrices,
\(\mathcal F_\nu(I)=I\), and it preserves \(\mathbb H_{d,0}\). For all
Hermitian \(A,B\),
\[
 \int\tr(AP_z)\tr(BP_z)\,\nu(dz)
 =\frac1d\tr\{A\mathcal F_\nu(B)\}.
\]
Consequently, the \(L^2(\nu)\) inner product of two one-observation densities
satisfies
\begin{equation}
 \int p_\rho(z)p_{\rho'}(z)\,\nu(dz)
 =d\tr\{\rho\mathcal F_\nu(\rho')\}.
\label{eq:density-overlap}
\end{equation}
If \(\nu\) is a frame-\(\varepsilon\) projective 2-design, then
\(\mathcal F_\nu\) is invertible, the dual shadow in
\eqref{eq:dual-shadow} is unbiased and has trace one, and, for traceless
Hermitian \(A,B\),
\begin{equation}
 \int p_{I/d+aA}(z)p_{I/d+bB}(z)\,\nu(dz)
 =1+dab\,\tr\{A\mathcal F_{\nu,0}(B)\}.
\label{eq:centered-density-overlap}
\end{equation}
If, in addition, \(d\ge2\), \(\rho\in\Theta_d^{\mathrm{loc}}(r)\),
\(dr\le1/4\), and \(O\in\mathbb H_{d,0}\), then
\begin{equation}
 \E_\rho\{\tr(OX_\nu)^2\}
 \le
 \frac{5}{4d}\tr\{O\mathcal F_{\nu,0}^{-1}(O)\}
 \le\frac54\frac{d+1}{d(1-\varepsilon)}\tr(O^2).
\label{eq:local-shadow-moment}
\end{equation}
\end{lemma}

\begin{proof}
The definition \eqref{eq:frame-operator} gives self-adjointness,
nonnegativity, and the first display. The first-moment normalization gives
\(\mathcal F_\nu(I)=I\). Since
\(\tr\{\mathcal F_\nu(A)\}=\tr(A)\), the trace-zero subspace is invariant.
Equation~\eqref{eq:density-overlap} follows from \eqref{eq:density}.

Condition~\eqref{eq:approx-frame} places the eigenvalues of
\(\mathcal F_{\nu,0}\) in
\([(1-\varepsilon)/(d+1),(1+\varepsilon)/(d+1)]\), proving invertibility.
Moreover,
\[
 \E_\rho X_\nu
 =\mathcal F_\nu^{-1}\{\mathcal F_\nu(\rho)\}=\rho.
\]
The block decomposition
\(\operatorname{span}\{I\}\oplus\mathbb H_{d,0}\) gives
\eqref{eq:dual-shadow}, and hence \(\tr(X_\nu)=1\).
Expanding \eqref{eq:density-overlap} around \(I/d\) gives
\eqref{eq:centered-density-overlap}.

For the last assertion, write \(H=\rho-I/d\). By \eqref{eq:density} and
\(\abs{\tr(HP_z)}\le\norm{H}_{\mathrm{op}}\),
\[
 \frac34\le p_\rho(z)=1+d\tr(HP_z)\le\frac54.
\]
Self-adjointness of \(\mathcal F_\nu^{-1}\) and the first identity give
\begin{align*}
 \E_\rho\{\tr(OX_\nu)^2\}
 &=\int p_\rho(z)
 \tr\{\mathcal F_{\nu,0}^{-1}(O)P_z\}^2\,\nu(dz)\\
 &\le\frac{5}{4d}
 \tr\{O\mathcal F_{\nu,0}^{-1}(O)\}\\
 &\le\frac54\frac{d+1}{d(1-\varepsilon)}\tr(O^2),
\end{align*}
which proves \eqref{eq:local-shadow-moment}.
\end{proof}

\section{Proof of the local upper bound}
\label{app:local-upper}

\begin{proof}[Proof of Theorem~\ref{thm:local}(i)]
Write \(H=\rho-I/d\) and \(E_i=X_i-\rho\). The Hoeffding decomposition of
\(h(X_1,X_2)=\tr(X_1X_2)\) gives
\begin{equation}
 \Var(\widehat T_2^U)
 =
 \frac{4\zeta_1}{n}
 +\frac{2\zeta_2}{n(n-1)},
\label{eq:hoeffding}
\end{equation}
where
\[
 \zeta_1=\E_\rho\{\tr(H E_1)\}^2,
 \qquad
 \zeta_2=\E_\rho\{\tr(E_1E_2)\}^2.
\]
This is the standard degree-two U-statistic variance decomposition
\citep{Hoeffding1948,Lee1990}. Here
\(\tr(E_1E_2)\) is the canonical second-order projection, so the coefficient
\(4/n\) in \eqref{eq:hoeffding} is exact. If \(\zeta_2\) were instead defined
as the variance of the raw kernel, the equivalent coefficient of \(\zeta_1\)
would be \(4(n-2)/\{n(n-1)\}\).

For traceless Hermitian \(O\), Lemma~\ref{lem:design-identities} implies
\[
 \Var_\rho\{\tr(OX_\nu)\}
 \le \E_\rho\{\tr(OX_\nu)^2\}
 \le\kappa_{\varepsilon_0}\tr(O^2),
 \qquad
 \kappa_{\varepsilon_0}=\frac{2}{1-\varepsilon_0}.
\]
It follows that
\[
 \zeta_1\le\kappa_{\varepsilon_0}\tr(H^2)
 \le\kappa_{\varepsilon_0}dr^2.
\]

Let \(\mathcal C_\rho\) be the covariance operator of \(E\) on the real
Hilbert space of traceless Hermitian matrices. The preceding bound gives
\(\norm{\mathcal C_\rho}_{\mathrm{op}}\le\kappa_{\varepsilon_0}\).
Independence and an
orthonormal expansion on this \((d^2-1)\)-dimensional space give
\[
 \zeta_2=\tr(\mathcal C_\rho^2)
 \le\kappa_{\varepsilon_0}^2(d^2-1).
\]
Substitution in \eqref{eq:hoeffding} yields
\begin{equation}
 \Var(\widehat T_2^U)
 \le
 \frac{4\kappa_{\varepsilon_0}dr^2}{n}
 +\frac{2\kappa_{\varepsilon_0}^2(d^2-1)}{n(n-1)}.
\label{eq:purity-variance}
\end{equation}

We next bound the deterministic approximation error. Write the eigenvalues of
\(\rho\) as \(\lambda_i=(1+x_i)/d\). Then
\(\sum_i x_i=0\) and
\[
 \max_i\abs{x_i}\le dr\le\frac14.
\]
For \(\phi(x)=(1+x)\log(1+x)\), Taylor's theorem on
\([-1/4,1/4]\) gives
\[
 \abs{\phi(x)-x-x^2/2}
 \le C\abs{x}^3.
\]
Since \(\sum_i x_i^2=d^2\tr(H^2)\),
\begin{align*}
 \abs{S(\rho)-\log d+\tfrac d2\tr(H^2)}
 &\le
 \frac{C\max_i\abs{x_i}}d
 \sum_i x_i^2\\
 &\le C d^3r^3.
\end{align*}
The untruncated quantity in \eqref{eq:local-estimator} has expectation
\(\log d-\tfrac d2\tr(H^2)\). Combining its squared bias with
\eqref{eq:purity-variance}, and using \(n(n-1)\ge n^2/2\), gives
\[
 \E_\rho
 \left[
 \log d-\frac d2\left\{\widehat T_2^U-\frac1d\right\}
 -S(\rho)
 \right]^2
 \le C_{\varepsilon_0}\left(
 \frac{d^3r^2}{n}+\frac{d^4}{n^2}+d^6r^6
 \right).
\]
Projection onto \([0,\log d]\) cannot increase the squared distance from
\(S(\rho)\in[0,\log d]\). This proves the upper bound.
\end{proof}

\section{Proof of the local lower bound}
\label{app:local-lower}

The local mixture construction uses the following consequence of
concentration on the unitary group
\citep[Theorem~5.17]{Meckes2019}. If a centered real-valued function on
\(U(d)\) is \(L\)-Lipschitz in Frobenius distance, then
\begin{equation}
 \{\E\abs{f(W)}^k\}^{1/k}
 \le C L\sqrt{\frac{k}{d}},
 \qquad k\ge1,
\label{eq:unitary-moments}
\end{equation}
where \(C\) is universal. Indeed, for \(d\ge3\), the cited concentration
inequality applied to \(f\) and \(-f\) gives
\[
 \Pr\{\abs{f(W)}\ge t\}
 \le2\exp\left\{-\frac{(d-2)t^2}{24L^2}\right\}.
\]
The identity
\(\E\abs{f(W)}^k=k\int_0^\infty t^{k-1}
\Pr\{\abs{f(W)}\ge t\}\,dt\), followed by the standard Gamma-function
bound, yields
\[
 \{\E\abs{f(W)}^k\}^{1/k}
 \le C_0L\sqrt{\frac{k}{d-2}}.
\]
For \(d\ge4\), this is bounded by the right-hand side of
\eqref{eq:unitary-moments} after enlarging the universal constant. For
\(d=2,3\), the Frobenius diameter bound
\(\abs{f(W)}\le2L\sqrt d\) for a centered Lipschitz function gives the same
conclusion, with another enlargement of \(C\).

\begin{proof}[Proof of Theorem~\ref{thm:local}(ii)]
Let
\[
 R_d=
 \begin{cases}
 \operatorname{diag}(1^{d/2},-1^{d/2}),&d\ \text{even},\\
 \operatorname{diag}(1^{(d-1)/2},-1^{(d-1)/2},0),&d\ \text{odd}.
 \end{cases}
\]
Then \(\tr R_d=0\), \(\norm{R_d}_{\mathrm{op}}=1\), and
\(\tr(R_d^2)\in\{d-1,d\}\).

We first derive the radius-dependent term. For \(t\in[r/2,r]\), define
\[
 \rho_t=\frac Id+tR_d.
\]
These states belong to \(\Theta_d^{\mathrm{loc}}(r)\) and are positive
definite because \(dr\le1/4\). Write
\(g(z)=\tr(R_dP_z)\). The one-observation density and its derivative are
\[
 p_t(z)=1+dtg(z),
 \qquad
 \dot p_t(z)=dg(z).
\]
The frame identity and \(\tr R_d=0\) give
\[
 \int g(z)^2\,\nu(dz)
 =\frac1d\tr\{R_d\mathcal F_{\nu,0}(R_d)\}
 \le\frac{1+\varepsilon}{d(d+1)}\tr(R_d^2)
 \le\frac{1+\varepsilon}{d+1}.
\]
Since \(\abs{g(z)}\le1\), the Fisher information in one observation satisfies
\[
 I(t)
 =\int\frac{d^2g(z)^2}{1+dtg(z)}\,\nu(dz)
 \le\frac{4(1+\varepsilon)d^2}{3(d+1)}
 \le3d.
\]

Let \(\psi(t)=S(\rho_t)\), and define
\[
 h(\beta)=\frac12\{(1+\beta)\log(1+\beta)
 +(1-\beta)\log(1-\beta)\}.
\]
In even dimensions, \(\psi(t)=\log d-h(dt)\); in odd dimensions, the entropy
deficit is \((d-1)h(dt)/d\). The derivative is negative on this interval.
Since
\(h'(\beta)=\operatorname{arctanh}(\beta)\ge\beta\), the even-dimensional
bound is \(\abs{\psi'(t)}\ge d^2r/2\). In odd dimensions, the factor
\((d-1)/d\), together with \(d\ge2\), yields the following uniform bound for
\(t\in[r/2,r]\):
\[
 \abs{\psi'(t)}\ge\frac{d^2r}{4}.
\]
On this interval, use the prior density
\[
 \pi_r(t)=\frac4r
 \cos^2\left\{\frac{2\pi}{r}\left(t-\frac{3r}{4}\right)\right\}.
\]
It vanishes at both endpoints and has Fisher information
\[
 I(\pi_r)=\int_{r/2}^r\frac{\{\pi_r'(t)\}^2}{\pi_r(t)}\,dt
 =\frac{16\pi^2}{r^2}.
\]
The densities \(p_t\) are positive and differentiable, their scores have mean
zero, and the prior removes the boundary term. The one-dimensional van Trees
inequality therefore gives
\begin{align}
 R_{n,d}^{\mathrm{loc}}(r;\nu)
 &\ge
 \frac{d^4r^2/16}{3nd+16\pi^2/r^2}
 \ge c_1\frac{d^3r^2}{n},
\label{eq:local-direction-lower}
\end{align}
where the last step uses \(r\ge n^{-1/2}\) and \(d\ge2\).

We next derive the term associated with the degenerate quadratic component.
Let \(\eta>0\) be a universal constant, to be chosen sufficiently small, such
that
\[
 2\eta\le1,
\]
and define
\[
 a_0=\sqrt{\frac{\eta}{n}},
 \qquad
 a_1=\sqrt{\frac{2\eta}{n}},
 \qquad
 \rho_{j,V}=\frac Id+a_jVR_dV^*,
\]
where \(V\) is Haar distributed. The assumptions on \(r\) imply
\(n\ge16d^2\). Since \(a_1\le n^{-1/2}\le r\) and
\(da_1\le dr\le1/4\), both priors are supported on
\(\Theta_d^{\mathrm{loc}}(r)\), and every state in their support is positive
definite.

Write \(\beta_j=da_j\).
The entropy under prior \(j\) is deterministic. Its deficit from \(\log d\)
is \(h(\beta_j)\) for even \(d\), and
\((d-1)h(\beta_j)/d\) for odd \(d\). Since
\(h'(\beta)=\operatorname{arctanh}(\beta)\ge\beta\) on \([0,1)\),
\begin{equation}
 \Delta_S:=\abs{S(\rho_{1,V})-S(\rho_{0,V})}
 \ge\frac{\eta d^2}{4n}.
\label{eq:entropy-separation}
\end{equation}

Let \(Q_j\) be the \(n\)-sample mixture law under prior \(j\), and let
\[
 q_j=\frac{dQ_j}{d\nu^{\otimes n}}
\]
be its density relative to the common reference measure \(\nu^{\otimes n}\).
For two independent Haar rotations \(V,V'\), write
\[
 A_V=VR_dV^*,
 \qquad
 K_\nu(V,V')
 =d\tr\{A_V\mathcal F_{\nu,0}(A_{V'})\}.
\]
Lemma~\ref{lem:design-identities} gives
\[
 \int q_jq_\ell\,d\nu^{\otimes n}
 =
 \E_{V,V'}\{1+a_ja_\ell K_\nu(V,V')\}^n.
\]
This kernel need not be a function of \(V^*V'\), because an approximate
frame need not be unitarily covariant. Nevertheless, it has the required
dimension-free moments. Conditional on \(V'\), it is centered as a function
of \(V\), since \(\E_V A_V=0\). Moreover, the frame bound gives
\[
 \norm{d\mathcal F_{\nu,0}(A_{V'})}_F
 \le\frac{d(1+\varepsilon)}{d+1}\norm{R_d}_F
 \le(1+\varepsilon)\sqrt d.
\]
Thus \(V\mapsto K_\nu(V,V')\) is at most
\(2(1+\varepsilon)\sqrt d\)-Lipschitz in Frobenius distance. Hence
\eqref{eq:unitary-moments}, followed by averaging over \(V'\), implies
\[
 \E\abs{K_\nu(V,V')}^k\le(C\sqrt k)^k
\]
uniformly in \(d\) and \(0\le\varepsilon<1\). In particular,
\(\E K_\nu=0\). Expanding the three pairwise
\(L^2(\nu^{\otimes n})\) inner products around the common reference measure
yields
\begin{align*}
 \int(q_1-q_0)^2\,d\nu^{\otimes n}
 &=
 \sum_{k=2}^n\binom nk
 \E\{K_\nu(V,V')^k\}(a_1^k-a_0^k)^2\\
 &\le
 \sum_{k=2}^n\binom nk
 \E\abs{K_\nu(V,V')}^k(a_1^k-a_0^k)^2\\
 &\le
 \sum_{k=2}^\infty
 \frac{n^k}{k!}(C\sqrt k)^k
 \left(\frac{2\eta}{n}\right)^k.
\end{align*}
The sum starts at \(k=2\) because \(\E K_\nu=0\); for odd \(k\), the
second line explicitly replaces the possibly signed moment by its absolute
moment.
The inequality \(k!\ge(k/e)^k\) bounds the last series by
\[
 \sum_{k=2}^\infty
 (2Ce\eta)^k k^{-k/2}
 \le C_1\eta^2
\]
for all sufficiently small universal \(\eta\). Fix \(\eta\) so that
\(C_1\eta^2\le1/4\). Then
\[
 \TV(Q_0,Q_1)
 \le\frac12
 \left\{\int(q_1-q_0)^2\,d\nu^{\otimes n}\right\}^{1/2}
 \le\frac14.
\]

Associate each value of an estimator with the nearer of the two entropy
values. An incorrect decision entails squared error at least
\(\Delta_S^2/4\),
while the sum of the two testing errors is at least
\(1-\TV(Q_0,Q_1)\). Averaging the Bayes risks under the two priors gives
\begin{equation}
 R_{n,d}^{\mathrm{loc}}(r;\nu)
 \ge\frac{\Delta_S^2}{8}\{1-\TV(Q_0,Q_1)\}
 \ge c_2\frac{d^4}{n^2},
\label{eq:local-mixture-lower}
\end{equation}
where \eqref{eq:entropy-separation} is used in the last step.
Combining \eqref{eq:local-direction-lower} and
\eqref{eq:local-mixture-lower}, and using
\(\max(x,y)\ge(x+y)/2\), proves part (ii).
\end{proof}

\section{Proof of the global lower bound}
\label{app:global-proof}

Let \(\nu_{\mathrm H}\) denote Haar probability measure on rank-one
projectors. We first bound the Fisher information in the reference submodel
\eqref{eq:global-submodel}.

\begin{lemma}
\label{lem:haar-fisher}
There is a universal \(d_0\) such that, for \(d\ge d_0\),
\(0<\mu\le1/9\), and \(\abs{t}\le\mu/2\), the Fisher information of one
Haar covariant observation satisfies
\[
 I_{\mathrm H}(t)\le\frac8d.
\]
\end{lemma}

\begin{proof}
Suppose first that \(d=2m\). For a Haar unit vector \(u\), define
\[
 B=\langle u,\Pi_bu\rangle,\qquad X=2B-1,\qquad
 a=1-2(\mu+t).
\]
Then \(B\sim\operatorname{Beta}(m,m)\), \(X\) is symmetric about zero, and
\(a\in[0,1]\). Since
\(\dot\rho_t=2(\Pi_s-\Pi_b)/d\), direct substitution into
\[
 I_{\mathrm H}(t)
 =
 d\,\E_u
 \frac{\langle u,\dot\rho_tu\rangle^2}
 {\langle u,\rho_tu\rangle}
\]
gives
\[
 I_{\mathrm H}(t)
 =
 4\E\frac{X^2}{1+aX}
 =
 4\E\frac{X^2}{1-a^2X^2}
 \le
 4\E\frac{X^2}{1-X^2}.
\]
The Beta integral satisfies
\[
 \E\frac1{B(1-B)}
 =
 2\,\frac{2m-1}{m-1}.
\]
Since \(X^2/(1-X^2)=1/\{4B(1-B)\}-1\), it follows that
\[
 \E\frac{X^2}{1-X^2}
 =
 \frac1{2(m-1)}
 =
 \frac1{d-2}.
\]
Thus \(I_{\mathrm H}(t)\le4/(d-2)\).

For \(d=2m+1\), let \(\Pi_s\) and \(\Pi_b\) have rank \(m\), let \(\Pi_0\)
have rank one, and set
\[
 \rho_t=
 \frac{2(\mu+t)}d\Pi_s+
 \frac{2(1-\mu-t)}d\Pi_b+
 \frac1d\Pi_0.
\]
Let \(C=\langle u,\Pi_0u\rangle\), and write the relative mass in the two
rank-\(m\) blocks as
\[
 \frac{\langle u,\Pi_bu\rangle}{1-C}=\frac{1+X}{2}.
\]
The Dirichlet decomposition implies that \(C\) and \(X\) are independent,
and \(X\) is again induced by a \(\operatorname{Beta}(m,m)\) variable.
Direct calculation gives
\[
 I_{\mathrm H}(t)
 =
 4\E
 \frac{(1-C)^2X^2}{1+(1-C)aX}.
\]
Condition on \(C\) and put \(y=1-C\). Since \(X\) is symmetric and
independent of \(C\),
\begin{align*}
 \E_X\frac{y^2X^2}{1+yaX}
 &=\frac12\E_X\left\{
 \frac{y^2X^2}{1+yaX}+\frac{y^2X^2}{1-yaX}
 \right\}\\
 &=\E_X\frac{y^2X^2}{1-y^2a^2X^2}
 \le \E_X\frac{X^2}{1-X^2},
\end{align*}
where the last inequality uses \(0\le y,a\le1\). Integrating over \(C\) and
using the preceding Beta calculation with \(d-1=2m\) gives
\[
 I_{\mathrm H}(t)\le4\E\frac{X^2}{1-X^2}=\frac4{d-3}.
\]
The two parity cases imply \(I_{\mathrm H}(t)\le8/d\) for all sufficiently
large \(d\).
\end{proof}

\begin{lemma}[Rotation averaging]
\label{lem:rotation-average}
Let \(\nu\) define a rank-one POVM as in Section~\ref{sec:model}, and let
\(\bar\rho_t\) denote either parity version of the reference submodel. For
\(V\in\mathsf U(d)\), define
\[
 \rho_{t,V}=V\bar\rho_tV^*,
 \qquad
 p_{t,V}(z)=d\,\tr(\rho_{t,V}P_z),
\]
and let \(I_{\nu,V}(t)\) be the Fisher information in \(t\) from one
observation. If \(\omega\) is Haar probability measure on
\(\mathsf U(d)\), then
\begin{equation}
 \int_{\mathsf U(d)}I_{\nu,V}(t)\,\omega(dV)
 =
 I_{\mathrm H}(t),
\label{eq:rotation-fisher}
\end{equation}
where \(I_{\mathrm H}(t)\) is the Haar Fisher information in
Lemma~\ref{lem:haar-fisher}. In particular, for every prior density \(\pi\)
on the parameter interval, there exists a deterministic \(V_0\) such that
\[
 \E_\pi I_{\nu,V_0}(t)
 \le
 \E_\pi I_{\mathrm H}(t).
\]
\end{lemma}

\begin{proof}
Since all states in the parameter interval are full rank,
\[
 I_{\nu,V}(t)
 =
 d\int
 \frac{\{\tr(V\dot{\bar\rho}_tV^*P_z)\}^2}
 {\tr(V\bar\rho_tV^*P_z)}
 \,\nu(dz).
\]
Fix \(z\) and write \(P_z=|v_z\rangle\langle v_z|\). If
\(V\sim\omega\), then \(u=V^*v_z\) is a Haar unit vector. The integral of
the preceding integrand over \(V\) is therefore
\[
 d\,\E_u
 \frac{\langle u,\dot{\bar\rho}_tu\rangle^2}
 {\langle u,\bar\rho_tu\rangle}
 =
 I_{\mathrm H}(t).
\]
The integrand is nonnegative, so Tonelli's theorem permits integration first
over \(V\) and then over \(z\). Since \(\nu\) is a probability measure,
\eqref{eq:rotation-fisher} follows. After integrating this identity against
\(\pi(t)\), the averaging argument gives an orientation \(V_0\) whose
prior-averaged Fisher information is no larger than the Haar average.
\end{proof}

\begin{lemma}
\label{lem:entropy-derivative}
For the even-dimensional submodel \eqref{eq:global-submodel}, let
\(\psi(t)=S(\rho_t)\). If \(0<\mu\le1/9\) and
\(\abs{t}\le\mu/2\), then
\[
 \psi'(t)
 =
 \log\frac{1-\mu-t}{\mu+t}
 \ge\frac12\log\frac1\mu.
\]
For the odd-dimensional submodel in the proof of
Lemma~\ref{lem:haar-fisher}, the derivative is multiplied by
\((d-1)/d\).
\end{lemma}

\begin{proof}
In even dimensions, differentiation gives
\[
 \psi'(t)
 =
 \log\frac{\lambda_b(t)}{\lambda_s(t)}
 =
 \log\frac{1-\mu-t}{\mu+t}.
\]
On the stated interval, the last ratio is at least \(1/(3\mu)\), and
\(\log\{1/(3\mu)\}\ge\tfrac12\log(1/\mu)\) for \(\mu\le1/9\).
The odd-dimensional formula follows because both nonneutral blocks have
rank \((d-1)/2\).
\end{proof}

\begin{proof}[Proof of Theorem~\ref{thm:global}]
Fix the measurement \(\nu\), and let \(\bar\rho_t\) denote the appropriate
parity version of the reference submodel. Set
\[
 \mu=2\sqrt{\frac dn}
\]
and take \(C\) in the theorem sufficiently large that
\(\mu\le1/9\) and Lemma~\ref{lem:haar-fisher} applies. On
\([-\mu/2,\mu/2]\), use the prior density
\[
 \pi_\mu(t)=\frac2\mu\cos^2\left(\frac{\pi t}{\mu}\right).
\]
It vanishes at the endpoints and has Fisher information
\[
 I(\pi_\mu)
 =
 \int\frac{\{\pi_\mu'(t)\}^2}{\pi_\mu(t)}\,dt
 =
 \frac{4\pi^2}{\mu^2}
 =
  \frac{\pi^2n}{d}.
 \]

By Lemmas~\ref{lem:haar-fisher} and~\ref{lem:rotation-average}, there exists a
deterministic orientation \(V_0\) such that
\[
 \E_{\pi_\mu}I_{\nu,V_0}(t)\le\frac8d.
\]
Fix this orientation and set \(\rho_t=V_0\bar\rho_tV_0^*\). For every \(t\)
in the support of \(\pi_\mu\), the state \(\rho_t\) is full rank. The density
\(p_t(z)=d\tr(\rho_tP_z)\) is strictly positive and continuously
differentiable in \(t\). Its score has mean zero because
\[
 \int \partial_t p_t(z)\,\nu(dz)
 =
 d\,\tr\left(
 V_0\dot{\bar\rho}_tV_0^*
 \int P_z\,\nu(dz)
 \right)
 =
 \tr(\dot{\bar\rho}_t)
 =
 0.
\]
The prior \(\pi_\mu\) is absolutely continuous and vanishes at both endpoints,
which justifies integration by parts in the van Trees inequality and removes
the boundary term.

Let \(\psi(t)=S(\rho_t)=S(\bar\rho_t)\), where the equality follows from
unitary invariance. The one-dimensional van Trees inequality gives, for every
estimator
 \(\widehat S\),
\[
 \sup_{\rho\in\Theta_d}
 \E_\rho\{\widehat S-S(\rho)\}^2
 \ge
 \frac{\{\E_{\pi_\mu}\psi'(t)\}^2}
 {n\E_{\pi_\mu}I_{\nu,V_0}(t)+I(\pi_\mu)}.
\]
The denominator is at most \((8+\pi^2)n/d\). By
Lemma~\ref{lem:entropy-derivative}, the squared numerator is bounded below by
a universal constant times \(\log^2(1/\mu)\); in odd dimensions, the
additional factor \((d-1)/d\) is bounded away from zero. Since
\[
 \log\frac1\mu
 =
 \frac12\log\left(\frac n{4d}\right),
\]
the claimed bound follows after taking the infimum over
\(\widehat S\in\mathcal E_n\).
\end{proof}

\end{appendix}

\bibliographystyle{imsart-nameyear}
\bibliography{refs}

\end{document}